\documentclass{llncs}
\usepackage{latexsym}
\usepackage{algorithmic}
\usepackage{amsmath}
\usepackage{amssymb}
\usepackage{cite}
\usepackage{graphicx}
\title{Approximate Point-to-Face Shortest Paths in $\mathcal{R}^3$\thanks{This
research was supported in part by NSF award CCF-0635013.}}

\author{Yam Ki Cheung and Ovidiu Daescu}
\institute{Department of Computer Science \\
       University of Texas at Dallas \\
       Richardson, TX 75080, USA \\
       {\tt \{ykcheung,daescu\}@utdallas.edu}}

\begin{document}

\pagestyle{empty}
\maketitle

\begin{abstract}
We address the {\em point-to-face} approximate shortest path problem in $\mathcal{R}^3$:
Given a set of polyhedral obstacles with a total of $n$ vertices, a source point $s$,
an obstacle face $f$, and a real positive parameter $\epsilon$, compute a path from $s$ to $f$
that avoids the interior of the obstacles and has length at most $(1+\epsilon)$ times the length
of the shortest obstacle avoiding path from $s$ to $f$.
We present three approximation algorithms that take
$O(n^4(L+\log(n/\epsilon))^2/\epsilon^2+n^2(L+\log(n/\epsilon))/\epsilon^3)$ time,
$O(T_{p-p}(n)*(1/\epsilon^2))$ time, and
$O(n^2\lambda(n)\log(n/\epsilon)/{\epsilon}^4+n^2\log(n\rho)\log(n\log\rho))$ time,
respectively, where $L$ is the precision of the integers used, $T_{p-p}(n)$ is the time
complexity of the point-to-point shortest path algorithm used, $\rho$ is the ratio of the length
of the longest obstacle edge to the
Euclidean distance between $s$ and $f$, and $\lambda(n)$ is a very slowly-growing function related to
the inverse of the Ackermann's function.
\end{abstract}

\section{Introduction}

The Euclidean shortest path problem among obstacles in the plane or space is one of the oldest and
 well-known problems in computational geometry. It has been intensively studied,
see~\cite{agarwal,chen,Clar87,har-peled,Her99,KMM97, Sar99,Pap85,sharir,Sharir86,GM87,KM88,OW88,Mit96,Mount84, Mit87},
as well as the survey by Mitchell \cite{Mit00}.
In general the problem is stated as: Given a set of obstacles in the d-dimensional
space $\mathcal{R}^d$ with a total of $n$ vertices, find a shortest (Euclidean) path between a source
point $s$ and a target point $t$ while avoiding the interior of the obstacles. There are two
commonly studied versions of this problem. One is the {\em single pair} version, that asks to find a
shortest path between two given query points. The other one is the
single source version, which first constructs a shortest path map with respect to a source
point $s$. After that, for any given query point $t$, a shortest path between $s$ and $t$ can be
found based on the shortest path map.

In $\mathcal{R}^2$, shortest paths are polygonal and turn only at the vertices of the polygonal obstacles.
Sharir and Schorr~\cite{Sharir86} have developed an $O(n^2\log n)$ time algorithm based on discrete
graph searching and the visibility graph of the obstacles, where $n$ is the number of the obstacle
vertices. Various studies, e.g.~\cite{GM87,KM88,OW88}, improved the time to quadratic in worst case.
Kapoor, Maheshwari, and Mitchell \cite{KMM97} gave an interesting $O(n+h^2\log n)$ time and $O(n)$
space algorithm based on the visibility graph approach, where $h$ is the number of holes (obstacles)
of the given input. This algorithm is the only algorithm known to be linear in $n$ in both time and space.
However, the time dependence on $h$ is quadratic, so the algorithm does not perform well if $h$ is not
relatively small compared to $n$. Avoiding the visibility graph approach, Mitchell~\cite{Mit96} developed
a version of the continuous Dijkstra method and obtained the first subquadratic, $O(n^{3/2+\epsilon})$
time algorithm. Subsequently, based on the same technique, this result was improved by Hershberger
and Suri~\cite{Her99} to $O(n\log n)$ time and $O(n \log n)$ space.

In $\mathcal{R}^3$, shortest paths among polyhedral obstacles are polygonal and turn only on obstacle
edges or vertices. However, unlike the case in $\mathcal{R}^2$, shortest paths need not lie on any
discrete graph. Sharir and Schorr~\cite{Sharir86} have shown that shortest paths in $\mathcal{R}^3$
are {\it geodesic}, i.e. paths must enter and leave an edge at the same angle. Given a distinct sequence
of edges, the local optimal path between two points can be unfolded at each edge to form a straight line,
and the local optimal path can be uniquely identified. Nevertheless, the problem is still significantly
harder than in $\mathcal{R}^2$. For algebraic considerations, Bajaj~\cite{Bajaj85, Bajaj88} has shown
that the algebraic complexity is exponential, since comparing the lengths of two paths may require
exponentially many bits.
Considering the combinatorial aspect of the problem, Canny and Reif~\cite{reif} have shown that the
shortest path problem in $\mathcal{R}^3$ is NP-hard.

An interesting special case of the shortest path problem is that in which the path is
restricted to the surface of a single polytope.
The first significant study of this special case in computational
geometry  is by Sharir and Schorr~\cite{Sharir86}. They gave an $O(n^3\log n)$ time
algorithm for convex polytopes by exploiting the special structure of geodesic paths along the surface
of a convex polytope. Mount \cite{Mount84} gave an improved algorithm for convex polytopes with running
time $O(n^2\log n)$. For general nonconvex polytopes, Mitchell, Mount and Papadimitriou~\cite{Mit87}
presented an $O(n^2\log n)$ algorithm extending the technique of Mount \cite{Mount84}.
See also~\cite{chen,agarwal,har-peled,sharir} for more studies for this topic.

Our real interest is in the general case in $\mathcal{R}^3$.
Several papers~\cite{Pap85,Clar87,Sar99} have presented polynomial time
approximation algorithms for computing an $\epsilon$-$approximate$ path, which has length at
most $(1+\epsilon)$ times the length of the shortest path between two query points,
where $\epsilon$ is a real parameter defining the quality of the approximation.
We will discuss those relevant to this paper in Section~\ref{prev}.

Throughout the paper we use the following notations.

\begin{tabbing}
$S$ denotes a set of polyhedral obstacles;\\
$E$ denotes the set of edges of $S$;\\
$P_{st}$ denotes a shortest path from $s$ to $t$ that avoids the interior of the obstacles; \\
$Q_{st}$ denotes an approximate path of $P_{st}$ (see next section); \\
$|P|$ denotes the length of a given path $P$; \\
$st$ denotes the line segment between points $s$ and $t$:\\
$|st|$ denotes the Euclidean distance between $s$ and $t$.\\
\end{tabbing}

\subsection{Our Results}

In this paper, we address the {\em point-to-face} shortest path problem in $\mathcal{R}^3$:
Given a set of polyhedral obstacles with a total of $n$ vertices, a source point $s$, and a destination
face $f$ of some obstacle in the input, compute an $\epsilon$-$approximate$ path from $s$ to $f$ that
avoids the interior of the obstacles.
We assume the obstacle surfaces are triangulated, so $f$ is a triangle in $\mathcal{R}^3$.


The problem has multiple applications, including robot navigation, path planning and
structural proteomics.
For example, in structural proteomics, after a protein surface has been segmented and pockets identified,
various descriptors can be associated
with key points of the pocket. Given a key point $s$, one can measure the pocket depth of $s$ by
computing various distance measures from
$s$ to the ``caps'' of the pocket, the most natural of which is the shortest collision free
distance from $s$ to the caps (a cap is a face of the convex hull of the protein).

We present three approximation algorithms for the point-to-face shortest path problem that take
$O(n^4(L+\log(n/\epsilon))^2/\epsilon^2+n^2(L+\log(n/\epsilon))/\epsilon^3)$ time,
$O(T_{p-p}(n)*(1/\epsilon^2))$ time, and
$O(n^2\lambda(n)\log(n/\epsilon)/{\epsilon}^4+n^2\log(n\rho)\log(n\log\rho))$ time,
respectively, where $L$ is the precision of the integers used, $T_{p-p}(n)$ is the time
complexity of the point-to-point shortest path algorithm used, $\rho$ is the ratio of the length
of the longest obstacle edge to the
Euclidean distance between $s$ and $f$, and $\lambda(n)$ is a very slowly-growing function related to
the inverse of the Ackermann's function. The main contribution of this paper is in the third algorithm.


\subsection{Previous Work in $\mathcal{R}^3$}
\label{prev}


In the approximate shortest path problem, an additional real positive parameter $\epsilon > 0$, which defines the quality of the approximation, is given as part of
the input, and the goal is to find a path between two given points $s$ and $t$
that avoids (the interior of)
the obstacles and has length at most $(1+\epsilon)$ times the length of the shortest obstacle avoiding
path between those two points. Such an approximate path is referred to as an $\epsilon$-$approximate$
path or an $\epsilon$-$approximation$ of the shortest path. In this paper, we assume
$\epsilon \le 1$.
There is also a more general approach, where given a point $s$, a shortest path map with respect to $s$
is constructed. In $\mathcal{R}^3$, this approach has been investigated for shortest paths on polytopes,
as well as for shortest paths among obstacles~\cite{Sar99}.

Papadimitriou \cite{Pap85} gave the first fully polynomial time approximation scheme for the general shortest
path problem in $\mathcal{R}^3$.
The time complexity of the algorithm is $O(n^4(L+\log(n/\epsilon))^2/\epsilon^2)$,
where $n$ is the complexity of the set of obstacles $S$, i.e. the number of edges, and $L$ is the precision of the
integers used, that is, the number of bits in the largest integer describing the coordinates of
any scene element.

\begin{figure}
    \begin{center}
    \leavevmode
    \includegraphics[height=2in]{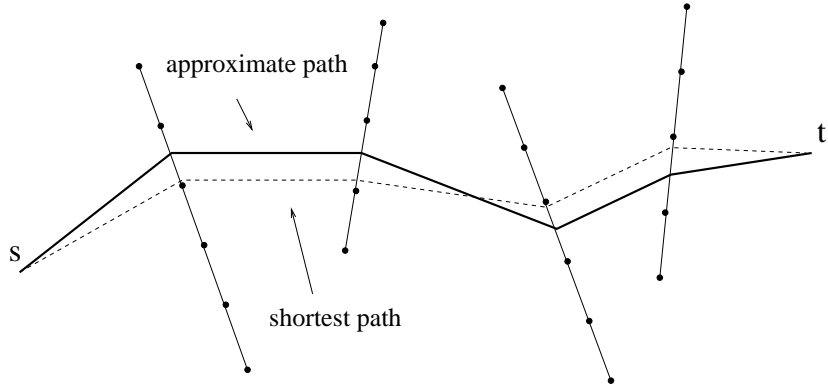}
    \caption{A shortest path and its approximation.}
    \label{jfig1}
    \end{center}
    \vspace{-0.15in}
\end{figure}

The approach is relatively simple and thus it could be implemented in practice. It discretizes the
problem by breaking every edge into a number of small segments. A visibility graph $G$ is constructed,
in which each node represents a segment. A link between two nodes is created if the two segments
represented are visible to each other. The weight of the link is set to the Euclidean distance between
the mid-points of the segments (see Fig.~\ref{jfig1} for an illustration).

To partition one edge $e$ of $S$, a coordinate system is chosen such that $e$ lies on the x-axis and the
origin is the closest point on $e$ to $s$. A sequence of points is added on $e$ with coordinates
$\{x_i=\epsilon_1a(1+\epsilon_1)^{i-1}:i=1,2,3,\ldots\}$, where $a$ is the Euclidean distance from
$s$ to $e$ and $\epsilon_1$ is a real positive parameter. As a result, the length of each segment is
no more than $\epsilon_1$ times the distance from $s$ to the segment. It is shown in \cite{Pap85}
that applying Dijkstra's algorithm on $G$ gives a path which is at most $(1+4n\epsilon_1)$ times the
length of the shortest path. Taking $\epsilon_1=\epsilon/4n$ gives an $\epsilon$-$approximate$ path.

Given a source point $s$, a real positive parameter $\epsilon$, and a set of polyhedral obstacles $S$ in $\mathcal{R}^3$, with a total of $n$ vertices,
Har-Peled \cite{Sar99} presented a technique to construct approximate shortest path maps in $O(T_{p-p}(n)*(1/\epsilon^2\log1/\epsilon))$ time on each face and in $O(T_{p-p}(n)*(1/\epsilon^2\log1/\epsilon)n^2)$ time total in $\mathcal{R}^3$, where $T_{p-p}(n)$ is the time complexity of the point-to-point shortest path
algorithm used. In \cite{Sar99}, the point-to-point shortest path
algorithm used is that of Clarkson\cite{Clar87}, which we will discuss later.
Once the map is constructed, it takes $O(\log n/\epsilon)$
time for each query, that is, given a query point $t\in\mathcal{R}^3$, the length of a
$(1+\epsilon)$-$approximate$ path
from $s$ to $t$ can be reported in $O(\log (n/\epsilon))$ time.

\begin{figure}
    \begin{center}
    \leavevmode
    \includegraphics[height=2in]{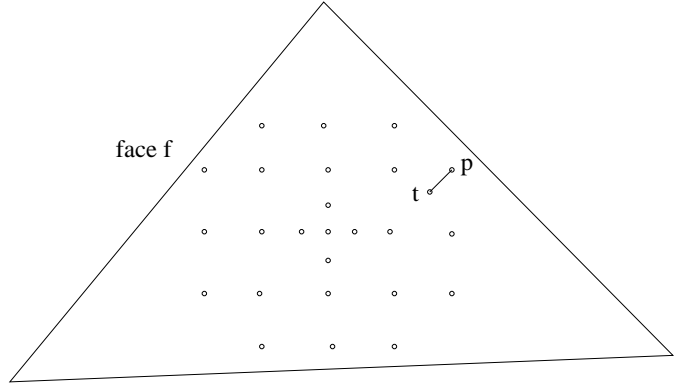}
    \caption{The distance function $f(t)$ is defined as $f(t)=w_{p}+|pt|$}
    \label{jfig2}
    \end{center}
\end{figure}

To construct a shortest path map for $s$
on an obstacle face $f$, Har-Peled's algorithm places a set of
$O(1/\epsilon^2\log1/\epsilon)$ points on $f$. A weighted Voronoi
diagram is then constructed on those points.
The weight of each point is the length of an $\epsilon/8$-$approximate$ path from $s$ to the point,
obtained by any
existing shortest path algorithm. For a point $t\in f$, the distance function $f(t)$ is defined as
$f(t)=w_{p}+|pt|$,
where $p$ is the point in the weighted Voronoi diagram closest
to $t$, $w_{p}$ is the
weight of $p$, and $|pt|$ is the Euclidean distance between $p$ and $t$ (see Fig.~\ref{jfig2}).
It was shown in \cite{Sar99} that
if the points are placed ``carefully'', $f(t)$ is at most $(1+\epsilon)$ times the length of the shortest
path between $s$ and $t$.

Clarkson~\cite{Clar87} gave an
$O(n^2\lambda(n)\log(n/\epsilon)/{\epsilon}^4+n^2\log (n\rho) \log(n\log\rho))$
algorithm for computing an $\epsilon$-$approximation$ of the shortest obstacle avoiding path
between two given points $s$ and $t$ in
$\mathcal{R}^3$, where $\rho$ is the ratio of the length of the longest obstacle edge to the
Euclidean distance between $s$ and $t$, and $\lambda(n)$ is a very slowly-growing function related to
the inverse of the Ackermann's function.

Let $S$ be the set of obstacles.
The general idea of Clarkson's approach is to construct a visibility graph $V$.
$V$ contains $s$, $t$, and the points on $S$ that $P_{st}$ might pass through. To limit the number of
edges per node in $V$, a cone structure $C$ is applied on every node $x\in V$
(see Fig.~\ref{jfig3}). In a cone $C_a\in C$,
with apex $x$, if $y$ is the closest node visible to $x$ then an edge $e$ between $x$ and $y$ is added to
$V$. The weight of $e$ is $|xy|$. Notice that there could be many nodes visible to $x$ in $C_a$, but at
most one edge is added. Clarkson showed that to ensure that the path found by running a single
source shortest path algorithm on $V$ is an $\epsilon$-$approximate$ path of $P_{st}$,
the size of the cone structure needed is $O(1/{\epsilon}^{d-1})$, where $d$ is the dimension of the problem
(i.e., $d=3$ in $\mathcal{R}^3$).
Hence, there are $O(1/{\epsilon}^d)$ edges incident to a node in $V$~\cite{Clar87}.
In the three-dimensional case, $P_{st}$ could make a turn at any point on any edge $e \in E$.
Let $a$ and $b$ be the end points of an edge $e \in E$. A point $x \in e$ can be expressed as a function
$x=\beta a+(1-\beta)b$, where $0 \le \beta \le 1$.
For a cone $C_a$ with apex $x=\beta a+(1-\beta)b$,
define the distance function $h^{C_a}(\beta)=|xy|$, where
$y$ is the closest node in $C_a$ visible to $x$.
One important observation Clarkson made is that $h^{C_a}$ is piecewise linear.
Given a predefined cone structure $C$, each $e\in E$ can be divided into segments according to
this piecewise
linearity. The set of line segments formed is called the \emph{combinatorial characterization} of
$V$ \cite{Clar87}, and is denoted as $K_{\epsilon}$. According to the study of Davenport-Schinzel
sequences \cite{SCK86}, the size of $K_{\epsilon}$ is $O(n^2\lambda(n)/{\epsilon}^2)$, where $n$ is the
number of obstacle edges, and $\lambda(n)$ is a very slowly-growing function (see above).
Hence the connectivity relations of $V$ can be represented by $K_{\epsilon}$.

\begin{figure}
    \begin{center}
    \leavevmode
    \includegraphics[height=2in]{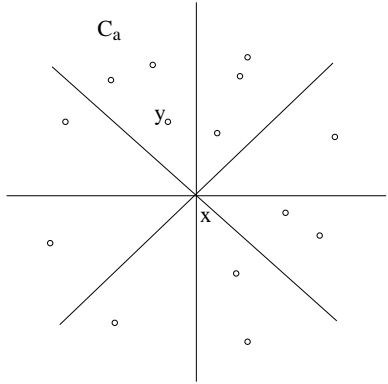}
    \caption{A simple cone structure with 8 cones.}
    \label{jfig3}
    \end{center}
\end{figure}

$V$ contains a set of
carefully selected Steiner points chosen as follows.
All endpoints of segments in $K_{\epsilon}$ are Steiner
points. Additional Steiner points are added to further divide each edge into segments no longer than
$\frac{\epsilon}{4n}B$, where $B$ is a lower bound on $|P_{st}|$.

The algorithm for computing an $\epsilon$-$approximate$ path consists of two phases.
In the first phase, $B$ is set to
$|st|$, $\epsilon$ is set to 1/2, and a simpler cone structure is used
to compute a $1/2$-$approximate$ path. This provides a better lower bound on $|P_{st}|$ for the second phase,
which gives the final approximation.


\section{Approximate Point-to-Face Shortest Paths}
In this section we
present three algorithms for finding approximate point-to-face shortest paths among polyhedral obstacles in $\mathcal{R}^3$.

\subsection{A General Approach}
The main idea used
is to place (a grid of) Steiner points on the target face $f$
such that there is at least one point close enough to $t$ to give us a good approximation.

Let $B$ and $D$ denote a lower bound and an upper bound on the shortest path between $s$ and $f$,
respectively.

Lemma \ref{2approximation} below is an application of Lemma~2.10 of \cite{Sar99}, which proved a more
general claim.
\begin{lemma}
\label{2approximation}
Let $x$ be any point on the target face $f$. Let $h$ be the point on $f$ closest to the source $s$
with respect to the Euclidean distance.
We have $|P_{sh}|\le 2 |P_{sx}|$.
\end{lemma}

\begin{proof}
Refer to Fig. \ref{jfig4}. Let $s'$ be the projection of $s$ on the plane supporting $f$.
If $s'\in f$, then $h=s'$, otherwise $h$ lies on the boundary of $f$.
It is obvious that $|hx|\le |sx|\le |P_{sx}|$. We have $|P_{sh}|\le |P_{sx}|+|hx| \le 2|P_{sx}|$.
\hfill $\Box$
\end{proof}

\begin{figure}
    \begin{center}
    \leavevmode
    \includegraphics[height=2in]{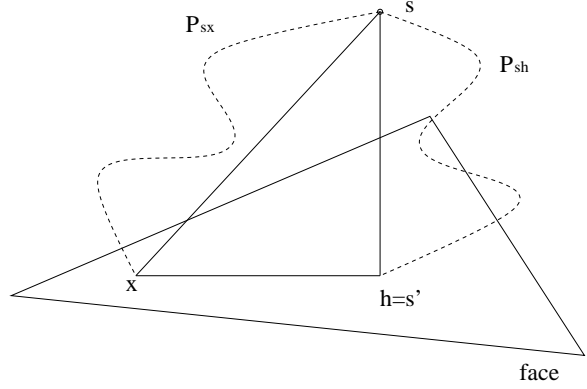}
    \caption{$|P_{sh}|$ is no more than $2|P_{sx}|$.}
    \label{jfig4}
    \end{center}
\end{figure}

Let $t\in f$ be the end point of the shortest path between $s$ and $f$. Hence the shortest path between $s$ and $f$ can also be expressed as $P_{st}$. By Lemma~\ref{2approximation}, we have $|P_{sh}|\le 2 |P_{st}|$. By computing an $\epsilon$-$approximate$ path between $s$ and $h$, i.e. $Q_{sh}$, we can obtain the lower bound and upper bound on $|P_{st}|$ as $B=|Q_{sh}|/(2(1+\epsilon))$ and $D=|Q_{sh}|$.

\subsection{The First Algorithm}
Our first algorithm for finding approximate point-to-face shortest paths extends Papadimitriou's algorithm for the point-to-point version. We proceed as follows.
The lower bound $B$ and the upper bound $D$ of the shortest path can be computed by identifying the
point $h$ and computing $Q_{sh}$ using Papadimitriou's algorithm.
Obviously $|ht|\leq P_{st}\leq D$, which implies $t$ must be within distance $D$ of $h$ (see Fig.~\ref{jfig5}).

\begin{figure}
    \begin{center}
    \leavevmode
    \includegraphics[height=2in]{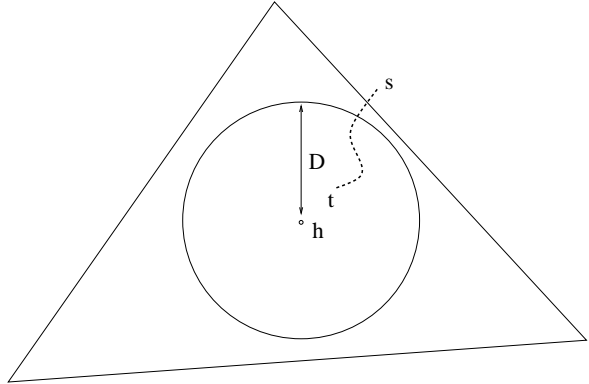}
    \caption{The point $t$ must be within $D$ distance of $h$.}
    \label{jfig5}
    \end{center}
\end{figure}

We define a sample grid as a
uniform grid of unit length $\epsilon B$, which is applied on the target face $f$ within distance $D$ of $h$. The total number
of sample points is $O(1/\epsilon^2)$.

\begin{lemma}
\label{sample}
There exists one sample point $t'$ that gives a $3\epsilon$-$approximate$ path.
\end{lemma}

\noindent \begin{proof}
Let $t'$ be the sample point closest to $t$. We have $|tt'| \le \epsilon B \le \epsilon |P_{st}|$ and
\begin{eqnarray*}
|P_{st'}| &\le& |P_{st}| + |tt'|\\
|Q_{st'}| &\le& (1+\epsilon)|P_{st'}| \\
        &\le& (1+\epsilon)(|P_{st}| + |tt'|)\\
        &\le& (1+\epsilon)(|P_{st}| +\epsilon |P_{st}|)\\
        &\le& (1+3\epsilon)|P_{st}|\\
\end{eqnarray*}
where $Q_{st'}$ is the approximating path and $P_{st}$ is the shortest path.
\hfill $\Box$
\end{proof}

In Papadimitriou's algorithm, each edge of $S$ is split into at most
$O((L+\log(1/\epsilon_1))/\epsilon_1)$ segments \cite{Pap85}.
The number of visibility graph edges resulting from this subdivision of edges of $S$ is
$O((n((L+\log(1/\epsilon_1))/\epsilon_1))^2)$. With $\epsilon_1=\epsilon/4n$ we have
$O(n^4(L+\log(n/\epsilon))^2/\epsilon^2)$ edges. We also have
$O(n^2(L+\log(n/\epsilon))/\epsilon)*(1/\epsilon^2))$ graph edges corresponding to visibility edges between
the segments on edges of $S$ and the sample points on $f$.
The number of vertices of the graph is $O(n^2(L+\log(n/\epsilon))/\epsilon+1/\epsilon^2)$. Dijkstra's algorithm takes
$O(N\log N +M)$ time, where $M$ is the number of graph edges and $N$ is the number of graph nodes. With $M= O(n^4(L+\log(n/\epsilon))^2/\epsilon^2+n^2(L+\log(n/\epsilon))/\epsilon^3)$ and $N=n^2(L+\log(n/\epsilon))/\epsilon+1/\epsilon^2)$, Dijkstra's algorithm takes $O(n^4(L+\log(n/\epsilon))^2/\epsilon^2+n^2(L+\log(n/\epsilon))/\epsilon^3)$ time to compute a $3\epsilon$-$approximate$ path between $s$ and $f$ by Lemma~\ref{sample}.

In order to have an $\epsilon$-$approximate$ path between $s$ and $f$, we simply use a new parameter $\epsilon ''=\epsilon/3$ to construct the visibility graph. The new parameter does not affect the complexity of the algorithm.

\begin{theorem}
Given a point $s$, a face $f$, a set $S$ of obstacles, with a total of $n$ vertices, and a real positive parameter $\epsilon$, an $\epsilon$-$approximation$ of
the shortest obstacle-avoiding path from $s$ to $f$ can be computed in
$O(n^4(L+\log(n/\epsilon))^2/\epsilon^2+n^2(L+\log(n/\epsilon))/\epsilon^3)$ time,
where $n$ is the complexity of $S$ and $L$ is the number of bits in the largest integer describing the coordinates of any scene element.
\end{theorem}

\subsection{The Second Algorithm}
Our second algorithm builds upon Har-Peled's point-to-point approximate shortest path algorithm. Har-Peled' algorithm constructs an approximate shortest path map on a given face $f$ with respect to a source $s$ in $O(T_{p-p}(n)*(1/\epsilon^2\log1/\epsilon))$ time, where $T_{p-p}(n)$ is the complexity of the point-to-point shortest path algorithm used. As mentioned previously, a full map on $f$ is unnecessary as $t$ must locate within $D$ radius of $h$. Hence only a partial map within $D$ radius of $h$ is sufficient to capture a good approximation. There is no actual query to be performed, since the sample points placed on $f$ to compute the weighted Voronoi diagram already serve the purpose as a sample grid. The number of sample points placed within $D$ radius of $h$ is $O(1/\epsilon^2)$ \cite{Sar99}. It takes $O(T_{p-p}(n)*(1/\epsilon^2))$ time to construct the partial map. If Clarkson's algorithm is used for $T_{p-p}$, we have:

\begin{theorem}
Given a query point $s$, a face $f$ and a real positive parameter $\epsilon$, an $\epsilon$-$approximation$ of
the shortest obstacle-avoiding path from $s$ to $f$ can be computed in $O((n^2\lambda(n)\log(n/\epsilon)/{\epsilon}^4+n^2\log (n\rho) \log(n\log\rho))*(1/\epsilon^2))$ time
, where $\rho$ is the ratio of the length of the longest obstacle edge to the
Euclidean distance between $s$ and $t$, and $\lambda(n)$ is a very slowly-growing function related to
the inverse of the Ackermann's function.
\end{theorem}

This second algorithm is very similar to the first one. Both algorithms need a sample grid of size $O(1/\epsilon^2)$ and execute a point-to-point shortest path algorithm on each point, which means we have to run a point-to-point shortest path algorithm $O(1/\epsilon^2)$ times.
In the next subsection, we show how to obtain a point-to-face shortest path algorithm which has the same asymptotic complexity as Clarkson's point-to-point version.

\subsection{The Third Algorithm}
\begin{theorem}
\label{main}
Given a query point $s$, a face $f$, a set of obstacles $S$ and a real positive parameter $\epsilon$, we can find an $\epsilon$-$approximate$ path of
the shortest obstacle-avoiding path from $s$ to $f$ in time
$O(n^2\lambda(n)\log(n/\epsilon)/{\epsilon}^4+n^2\log (n\rho) \log(n\log\rho))$, where $\rho$ is the ratio of the length of the longest obstacle edge to the
Euclidean distance between $s$ and $t$, and $\lambda(n)$ is a very slowly-growing function related to
the inverse of the Ackermann's function.
\end{theorem}
Notice that in Theorem~\ref{main} the point-to-face approximation has the same complexity as Clarkson's
point-to-point approximation~\cite{Clar87}. Indeed, our solution builds upon Clarkson's point-to-point solution.

We will use a different approach to prove Theorem~\ref{main}.
Instead of building a sample grid on $f$ first, we will construct a visibility graph first and then
decide the additional Steiner points needed on $f$ based on the visibility graph constructed.
All additional Steiner points on $f$ can be added directly to the visibility graph.
Hence, we only need to execute Dijkstra's algorithms once for all Steiner points placed on $f$.
We will explain how additional Steiner points and the corresponding visibility graph edges are produced, and give a count on the total number of Steiner points.

\begin{figure}
    \begin{center}
    \leavevmode
    \includegraphics[height=2in]{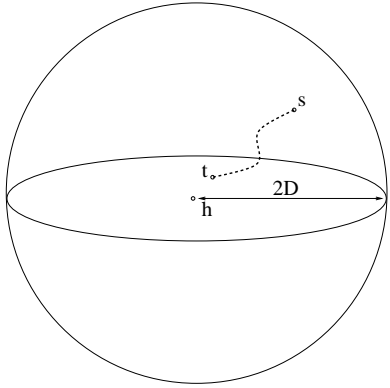}
    \caption{$P_{st}$ cannot travel beyond the sphere  centered at $h$ with radius $2D$.}
    \label{jfig6}
    \end{center}
\end{figure}

Clarkson's algorithm consists of two phases. In the first phase, a coarse approximation of
the shortest path between the source and destination is computed, which gives a better lower bound
for the second phase. We set $h$, which is the closest point on $f$ to $s$ in Euclidean distance,
as the destination point and find a $\frac{1}{2}-approximation$ of the shortest path between $s$
and $h$ using
the first phase of Clarkson's approach. Let $t\in f$ be the end point of the shortest path between
$s$ and $f$. Set the length of the $\frac{1}{2}-approximate$ path as the upper bound $D$ of $P_{st}$
and let $B=D/2(1+1/2)$ as the lower bound, following Lemma.~\ref{2approximation}. Obviously, $P_{st}$ cannot
travel beyond the sphere centered at $h$ with radius $2D$ (see Fig.~\ref{jfig6}). In the second phase,
following Clarkson's approach, we only need to partition obstacle edges within the sphere into
segments of length $\epsilon B/4n$ and then apply the cone structure on each Steiner point created.
With this, we have the following lemma.

\begin{lemma}
\label{lemma-edge}
If $t$ is on an edge $e \in f$ then there exists a Steiner point $t' \in e$ such that
$|Q_{st'}| \le (1+5/4 \epsilon) |P_{st}|$.
\end{lemma}

\begin{proof}
On $e$, choose the Steiner point closest to $t$ as $t'$.
See Fig.~\ref{jfig7} for an illustration.
Since each segment is no longer than
$\frac{\epsilon}{4n}|P_{st}|$, $|tt'|\leq \frac{\epsilon}{8n}|P_{st}|$.
By triangle inequality, we have $|P_{st'}|\leq |P_{st}|+|tt'|$ and
\\
\begin{eqnarray*}
|Q_{st'}|       &\leq& (1+\epsilon)|P_{st'}|~~~~~~(from\ Clarkson's\ Steiner\ point\ placement) \\
                &\leq& (1+\epsilon)(|P_{st}|+|tt'|) \\
                &\leq& (1+\epsilon)|P_{st}|+(1+\epsilon)(\frac{\epsilon}{8n}|P_{st}|) \\
                &\leq& (1+\epsilon+\frac{\epsilon}{8n}+\frac{{\epsilon}^2}{8n})|P_{st}|\\
                &\leq& (1+5/4\epsilon)|P_{st}|~(since \ \epsilon \leq 1 \ and \ n \ge 1)
\end{eqnarray*}
\hfill $\Box$

\begin{figure}
    \begin{center}
    \leavevmode
    \includegraphics[height=2in]{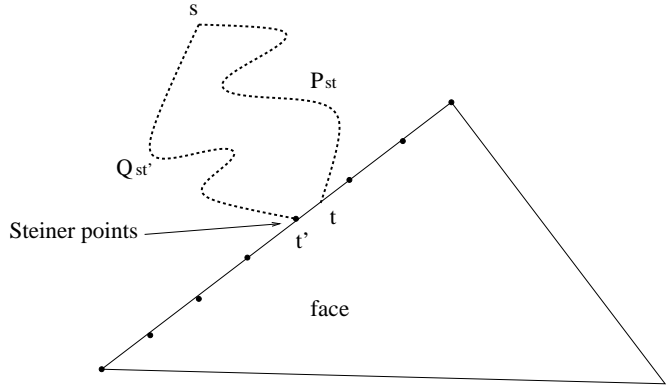}
    \caption{The Steiner point $t'$ gives a good approximation for $t$.}
    \label{jfig7}
    \end{center}
\end{figure}
\hfill $\Box$
\end{proof}

To handle the case when $t$ is in the interior of $f$ we need to add two sets of Steiner points.
We do this as follows.
For the first set, for each obstacle edge $e$, we add point $u\in e$ as a new Steiner point, if $u$ satisfies:\\
1) its projection $u'$ in the plane containing $f$ is in the interior of $f$, and \\
2) $u$ is visible from $u'$, and\\
3) $uu'$ is tangent to some obstacle, which does not contain $e$. \\
We apply the cone structure (following Clarkson's algorithm) on $u$ and add the corresponding edges to the visibility graph. Since there are at most $n^2$ such points and the number of Steiner points in the original visibility graph is $\Omega(n^2)$, the complexity of the visibility graph is unchanged.

To obtain the second set of Steiner points, we proceed as follows.
For each existing Steiner point $v$ on an obstacle edge, let $v'$ be the projection of $v$ in the plane containing $f$. We add $v'$ and edge $v'v$ to the visibility graph if $v'$ is in the interior of $f$ and $v$ is visible from $v'$. Note that we do not apply the cone structure on $v'$. Observe that, we at most double the number of Steiner points
and each additional Steiner point introduces exactly one edge. The complexity of the visibility graph remains the same. We will discuss later how to efficiently determine the visibility between a point and its projection on a plane.

\begin{lemma}
\label{lemma-per}
Let the contact point of $P_{st}$ on the last obstacle edge before reaching $f$ be $c_1$, that is,
$c_1t$ is the last segment of $P_{st}$. If $t$ is an interior point of $f$ then
$c_1t$ must be perpendicular to $f$.
\end{lemma}

\begin{proof}
We make the proof by contradiction. Suppose $t$ is in the interior of $f$ and $c_1t$ is
not perpendicular to $f$. We will construct a new path from $c_1$ to some point $t'\in f$
such that the new path is shorter than $|c_1t|$.

Let $c_1'$ be the projection of $c_1$ on the plane $\pi$ containing $f$. Consider the plane
formed by $c_1$, $c_1'$ and $t$. Let $t'=t$. We shift $t'$ along the segment $tc_1'$ until either $c_1'$ is
reached or the segment $c_1t'$ intersects some edge $e$ at a point $c_2$. We have $|c_1t'|<|c_1t|$, a contradiction.

\hfill $\Box$
\end{proof}
\begin{figure}
    \vspace*{-0.35in}
    \begin{center}
    \leavevmode
    \includegraphics[height=2in]{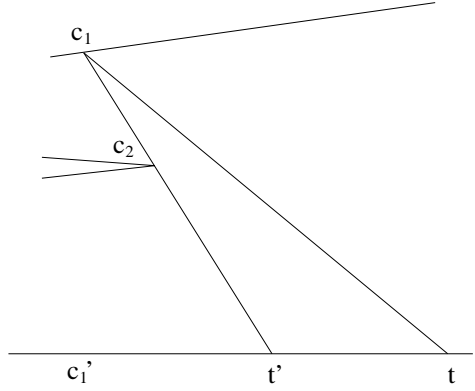}
    \caption{Illustration of (the proof of) Lemma~\ref{lemma-per}.}
    \label{fig2}
    \end{center}
    \vspace*{-0.35in}
\end{figure}

\begin{lemma}
\label{lemma-interior}
If $t$ is an interior point of $f$, in the modified visibility graph there exists at least one Steiner
point $t'$ on $f$ and such that $|Q_{st'}|\le(1+5/4\epsilon)|P_{st}|$.
\end{lemma}

\begin{proof}



By $Lemma \ \ref{lemma-per}$, the last segment $c_1t$ of $P_{st}$ is perpendicular to $f$, since $t$ is an interior
point of $f$. Assume $c_1$ is not a Steiner point, since otherwise we are done.
On the edge containing $c_1$, let $d_1$ be one of the two Steiner points neighboring $c_1$,
specifically, one visible from its projection on $f$. Such a Steiner point exists from our placement of Steiner points (set one). We have $|c_1d_1| \leq \frac{\epsilon}{4n}|P_{st}|$.
Let the projection of $d_1$ on $f$ be $t'$. Use an alternative path from $s$ to $t'$ to approximate $P_{st}$ (See Fig.~\ref{jfig8} for an illustration). We have


\begin{eqnarray*}
|Q_{st'}| &=& |Q_{sd_1}| + |d_1t'| \\
            &\leq& |P_{sd_1}| + \epsilon |P_{st}| + |d_1t'|\\
            &\leq& |P_{sd_1}| + \epsilon |P_{st}| + |c_1t| + |c_1d_1|\\
            &\leq& |P_{st}|+\epsilon |P_{st}|+\frac{\epsilon}{4n}|P_{st}|\\
            &\leq& (1+5/4\epsilon) |P_{st}| \\
\end{eqnarray*}

\hfill $\Box$
\end{proof}

\begin{figure}
    \begin{center}
    \leavevmode
    \includegraphics[height=2in]{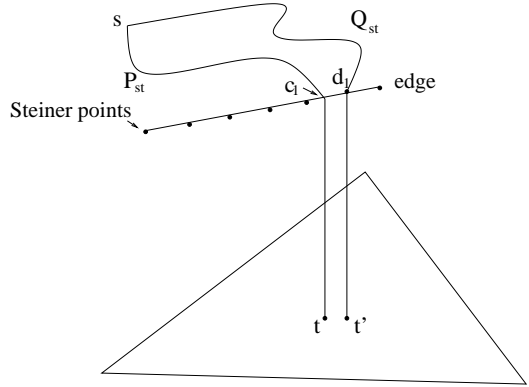}
    \caption{Approximation of $P_{st}$ by an alternative path.}
    \label{jfig8}
    \end{center}
\end{figure}


Following $Lemma~\ref{lemma-edge}$ and $Lemma~\ref{lemma-interior}$, the new visibility graph gives us a
$5/4\epsilon$-$approximate$ path from $s$ to $f$.
To conclude the proof of Theorem~\ref{main} we need to show how to find the additional Steiner points
and visibility edges within the given time bound.

\subsubsection{Visibility Computation}

We now show how to compute the additional Steiner points on obstacle edges as well as on the target obstacle face $f$. Since addition of Steiner points does not change the complexity of the visibility graph, the cost of the shortest path computation is $O(n^2\lambda(n)\log(n/\epsilon)/{\epsilon}^4+n^2\log(n\rho)\log(n\log\rho))$.
We would like to keep the computation of additional Steiner points and edges that capture
the visibility between Steiner points on obstacle edges and their projections on $f$ within this time bound.

A naive approach to compute the visibility between two points $u$ and $u'$
would require $O(n)$ time by checking whether the line segment $uu'$ intersects any obstacle.
Over all points, this will exceed the time bound above.

We can try to reduce the time to compute the visible pairs by using a ray shooting data structure.
As mentioned in~\cite{Shar03}, the general ray shooting problem in three dimensions is still far
from being fully solved.
However, our problem is a special case of the ray shooting problem. We need to determine the visibility
between a point and its orthogonal projection on the target plane $\pi$ supporting $f$, instead of the
visibility between any two arbitrary points, i.e. the direction for ray shooting is always perpendicular
to $\pi$. On each obstacle edge $e$, all Steiner points and their projections on $\pi$ are coplanar.
Let $\pi'$ be the plane passing through $e$ and perpendicular to $\pi$. By computing
$G=\pi' \bigcap S$ for each edge of $S$ (see Fig.~\ref{Fig-2d}), we can transform the three dimensional
visibility problem
into $O(n)$ subproblems in dimension two, where $n$ is the number of obstacle edges.

\begin{figure}
    \begin{center}
    \leavevmode
    \includegraphics[height=2in]{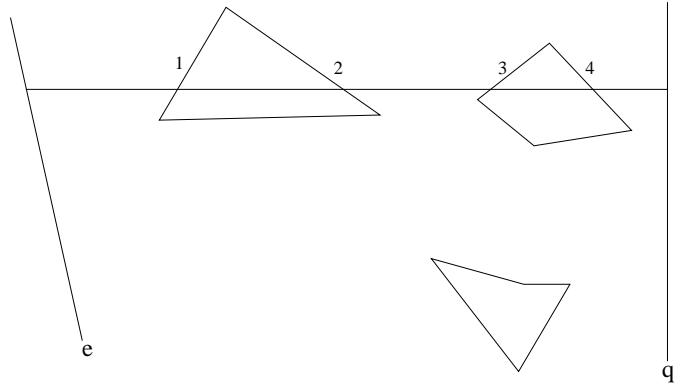}
    \caption{The sweep line intercepts edges $e,1,2,3,4,$ and $q$}
    \label{Fig-2d}
    \end{center}
\end{figure}

Let edge $e'$ be the projection of an edge $e$ onto $\pi$.
Notice that all line segments between Steiner points on $e$
and their projections on $\pi$ (that is, on $e'$) are parallel. We can perform a plane sweeping on $G$
with a horizontal line and maintain a data structure which stores edges intersecting with the sweep line
at a certain time instance. The data structure is updated when
the sweep line passes a vertex and
each insertion and deletion precess requires $O(\log n)$ time~\cite{shamos}. During this process,
we could partition
$e$ into $O(n)$ open segments, such that all points on each segment have the same visibility answer
with respect to their corresponding
projections on $e'$. Furthermore, the end points of segments that are visible from their projection on $f$ are the new Steiner points to be added on obstacle edges. The total time of adding Steiner points is $O(n^2\log n+N)$, where $N$ is the number of Steiner points in the original visibility graph.

We can simplify this process as follows.
Since the order of the edges intersecting with the sweep line is not important, we
can avoid the plane sweep procedure altogether while still partitioning each edge into $O(n)$ segments
in $O(n\log n)$ time. Recall that we assume the obstacles are triangulated. For each triangle $\Delta$, find the
intersection segment $r=\pi' \bigcap \Delta$. If at least one vertex of $r$ is inside of the quadrilateral
formed by the endpoints of $e$ and $e'$, we project $r$ on $e$. The projection of $r$ on $e$ is the
region in which the visibility of Steiner points to their counterparts are blocked by $r$. Hence we can
divide $e$ into $O(n)$ segments in $O(n\log n)$ time. Again, the total time is $O(n^2\log n+N)$.

Note that the cost of visibility computation is dominated by the cost of the shortest path computation, i.e. the overall time to approximate the shortest path from $s$ to $f$ is $O(n^2\lambda(n)\log(n/\epsilon)/{\epsilon}^4+n^2\log(n\rho)\log(n\log\rho))$.

This concludes the proof of Theorem~\ref{main}.

\section{Conclusions}
In this paper we discussed three point-to-face approximate shortest path algorithms
in $\mathcal{R}^3$.
It is interesting to notice that the point-to-face shortest
path must end on $f$ within a certain range from $h$, where $h$ is the point on $f$ closest to $s$ in
Euclidean distance. That range can be obtained by applying known point-to-point shortest path algorithms
between $s$ and $h$. After placing a sample grid near $h$, we proved we
can find an $\epsilon$-$approximate$ path
between $s$ and $f$ in $O(n^4(L+\log(n/\epsilon))^2/\epsilon^2+n^2(L+\log(n/\epsilon))/\epsilon^3)$ time
by extending Papadimitriou's algorithm or in
$O((n^2\lambda(n)\log(n/\epsilon)/{\epsilon}^4+n^2\log (n\rho) \log(n\log\rho))*(1/\epsilon^2))$ time by
extending Har-Peled's algorithm. However, this approach still requires to execute a point-to-point
shortest path algorithm $O(1/\epsilon^2)$ times.
Finally we showed that Clarkson's point-to-point shortest path approach can be extended to solve the
problem by adding additional Steiner points directly to the visibility graph, without changing the
asymptotic complexity of the algorithm, resulting in an
$O(n^2\lambda(n)\log(n/\epsilon)/{\epsilon}^4+n^2\log(n\rho)\log(n\log\rho))$ time algorithm for finding
an $\epsilon$-$approximate$ path between a source point $s$ and an obstacle face $f$.

\end{document}